\documentclass[11pt,a4paper,oneside,onecolumn]{article}

\usepackage[english]{babel}
\usepackage[ansinew]{inputenc}
\usepackage{amssymb}
\usepackage{amsmath}
\usepackage{amsthm}
\usepackage{amsfonts}
\usepackage{color}
\usepackage{subfig}
\usepackage{pseudocode}
\usepackage{hyperref}
\usepackage{framed}
\usepackage{graphicx}

\newtheorem{theorem}{Theorem}
\newtheorem{lemma}[theorem]{Lemma}

\newtheorem{MIdef}[theorem]{Definition}

\DeclareMathOperator{\turn}{turn} %
\DeclareMathOperator{\alt}{alt} %
\DeclareMathOperator{\ball}{B} %

\title{Drawing the double circle on a grid of minimum size}

\author{
S. Bereg\thanks{Department of Computer Science, University of Texas at Dallas, USA. Email: besp@utdallas.edu.}
\and
R. Fabila-Monroy\thanks{CINVESTAV, Instituto Polit\'ecnico Nacional, Mexico. Email: ruyfabila@math.cinvestav.edu.mx. Partially supported by grant 153984 (CONACyT, Mexico).}
\and
D. Flores-Pe\~naloza\thanks{Departamento de Matem\'aticas, Facultad de Ciencias, UNAM, Mexico. Email: dflorespenaloza@gmail.com.
Partially supported by grants 168277 (CONACyT, Mexico) and IA102513 (PAPIIT, UNAM, Mexico).}
\and
M. A. Lopez\thanks{Department of Computer Science, University of Denver, USA. Email: mlopez@du.edu.}
\and
P. P\'erez-Lantero\thanks{Escuela de Ingenier\'ia Civil en Inform\'atica, Universidad de Valpara\'{i}so, Chile. Email: pablo.perez@uv.cl. Partially supported by grant CONICYT, FONDECYT/Iniciaci\'on 11110069 (Chile).}
}

\begin{document}
\maketitle

\begin{abstract}
In 1926, Jarn\'ik introduced the problem of drawing a convex $n$-gon 
with vertices having integer coordinates. He constructed such a 
drawing in the grid $[1,c\cdot n^{3/2}]^2$ for some constant $c>0$, and showed 
that this grid size is optimal 
up to a constant factor. We consider the analogous problem for 
drawing the double circle, and prove that it can be done within
the same grid size. 
Moreover, we give an $O(n)$-time algorithm to 
construct such a point set. 
%
\end{abstract}


\section{Introduction}

Given $n\geq 3$, a {\em double circle} is a set
$P=\{p_0,p_1,\ldots,p_{n-1},p'_0,p'_1,\ldots,p'_{n-1}\}$
of $2n$ planar points in general position such that:
(1) $p_0,p_1,\ldots,p_{n-1}$ are precisely the vertices of the
convex hull of $P$ labelled in counterclokwise order around the boundary;
(2) point $p'_i$ is close to
the segment joining $p_i$ with $p_{i+1}$;
(3) the line passing through $p_i$ and $p'_i$ separates
$p_{i+1}$ from $P$; and
(4) the line passing through $p'_i$ and $p_{i+1}$ separates
$p_{i}$ from $P$ (see Figure~\ref{fig:double-circle}).
Subindices are taken modulo $n$.
The double circle has been considered in combinatorial geometry and it is conjectured
to have the least number of triangulations~\cite{aichholzer2004,aichholzer2008}.

\begin{figure}[h]
	\centering
	\includegraphics[scale=0.6]{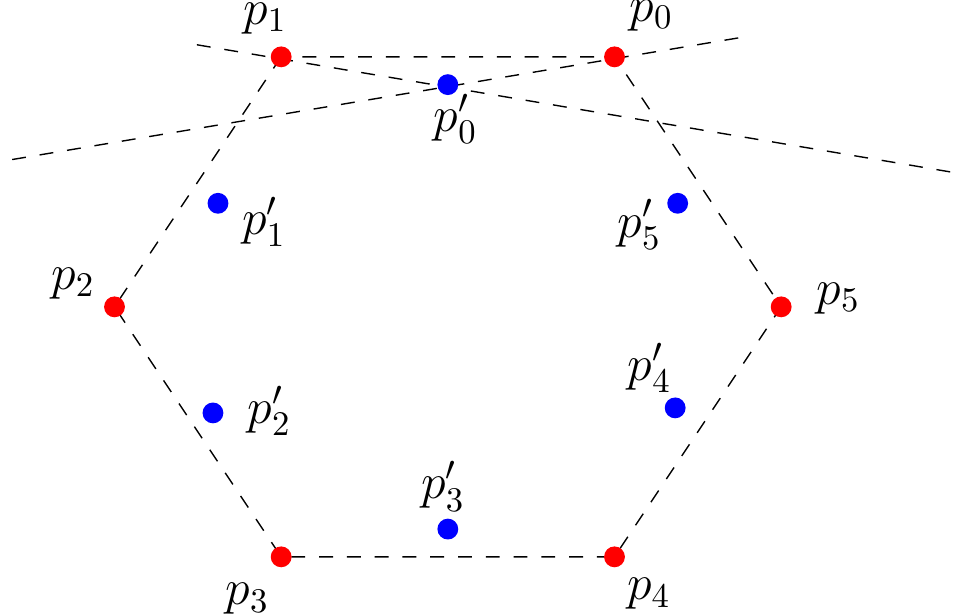}
	\caption{\small{A double circle of twelve points.}}	
	\label{fig:double-circle}
\end{figure}

Drawing an $n$-vertex convex polygon with integer vertices can be easily done
by considering the $n$ points $(1,1),(2,4),(3,9),\ldots,(n,n^2)$ as the vertices
of the polygon. In this case the {\em size} of the integer point set is equal
to $n^2-1=\Theta(n^2)$, where size refers to the smallest $N$ such that the point set
can be translated to lie in the grid $[0,N]^2$. In 1926, Jarn\'ik~\cite{jarnik1925}
showed how to draw an $n$-vertex convex polygon with size $N=O(n^{3/2})$ and proved that
this bound is optimal. In recent years the so-called Jarn\'ik polygons and extensions
of them have been studied~\cite{barany2012,martin2003}.

Given any integer point $(i,j)$, we say that $(i,j)$ is {\em visible} (from the origin)
if the interior of the line segment joining the origin and $(i,j)$ contains no
lattice points. Observe that $(i,j)$ is visible
if and only if $\gcd(i,j)=1$,
where $\gcd(i,j)$ denotes the greatest common divisor of $i$ and $j$.
We consider points as vectors as well, and vice versa.
A Jarn\'ik polygon is formed by choosing a natural number $Q$, and taking
the set $V_Q$ of visible vectors $(i,j)$ such that
$\max\{|i|,|j|\}\leq Q$~\cite{huxley1996,jarnik1925,martin2003}. The polygon
is then the unique (up to translation) convex polygon whose edges, viewed as vectors, are precisely
the elements of $V_Q$, that is, the vertices can be obtained by starting
from an arbitrary point and adding the vectors of $V_Q$, one by one,
in counterclockwise order, to the previously computed vertex (see Figure~\ref{fig:jarnik-example}).

\begin{figure}[h]
	\centering
	\includegraphics[scale=0.4]{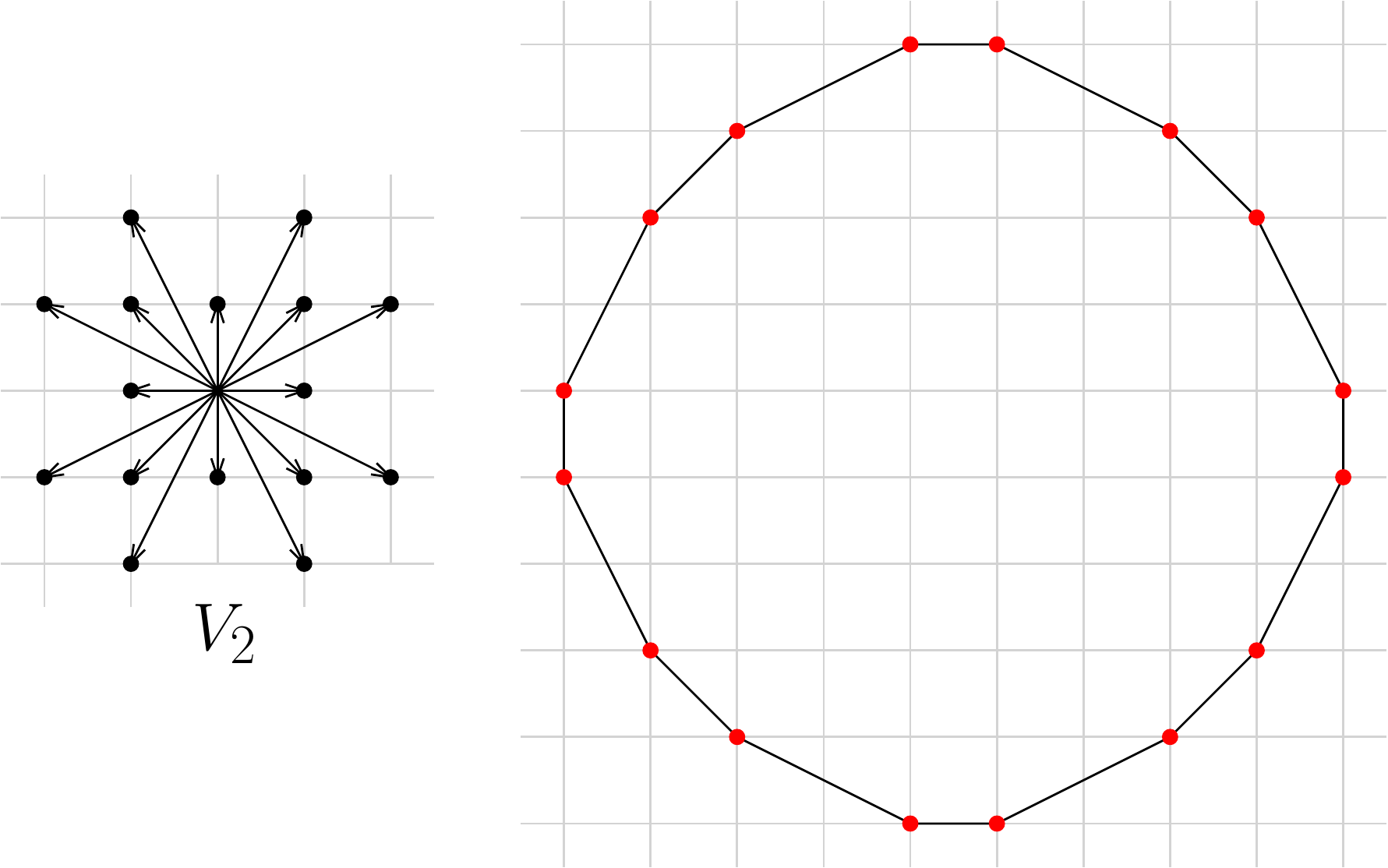}
	\caption{\small{A Jarn\'ik polygon (right) and its generating vectors $V_2$ (left).}}		
	\label{fig:jarnik-example}
\end{figure}

We study how to draw a $2n$-point double circle with integer points
using the smallest size $N$. We present an $O(n)$-time algorithm
that correctly constructs the double circle with size within $O(n^{3/2})$, where
that bound is also optimal. We consider the unit cost RAM model, where all
operations with real numbers, including the floor and ceiling functions, require constant time.
In Section~\ref{sec:algorithm} we show our algorithm, 
and in Section~\ref{sec:proof} its correctness is proved. Finally,
in Section~\ref{sec:conclusions}, we state future work.
Some examples are given in Appendix~\ref{sec:examples}.

\section{Double circle construction}\label{sec:algorithm}

Observe that a simple construction with quadratic size is
as follows: Consider the function
$f(x)=x^2+x$. For $i=1,\ldots,2n-1$, add the point $(i,f(i))$ if $i$ is odd, and the point
$(i,f(i)+2)$ otherwise. The final point is $(n,\frac{f(2n-1)+f(1)}{2}-1)=(n,2n^2-n)$, i.e., the
point just below the midpoint of the segment connecting $(1,f(1))$ and $(2n-1,f(2n-1))$.
The size of the resulting point set is $N=f(2n-1)-f(1)=(2n-1)^2+(2n-1)-2=4n^2-2n-2=\Theta(n^2)$.

We say that a sequence $V$ of vectors is {\em symmetric} if $V$ contains an even number
of vectors sorted counterclockwise around the origin, and for every vector $a$ in $V$ its
opposite vector $-a$ is also in $V$. Observe that any sequence of vectors defining a
Jarn\'ik polygon is symmetric. For any sequence $V=[v_1,v_2,\ldots,v_{2t}]$ of $2t$ vectors 
let the point set $\mathcal{P}(V):=\{p_1,p_2,\ldots,p_{2t}\}$, where $p_1=v_1$ and
$p_i=p_{i-1}+v_i$ for $i=2,\ldots,2t$. Note that if we sort the
elements of $V_Q$ around the origin then the elements of $\mathcal{P}(V_Q)$ are the
vertices of the Jarn\'ik polygon. Furthermore, if $V$ is symmetric
then the elements of $\mathcal{P}(V)$ are in convex position.
Let sequence $\alt(V):=[v_2,v_1,v_4,v_3,\ldots,v_{2t},v_{2t-1}]$
(see Figure~\ref{fig:alternate-example2} for an example with $t=8$).
For any scalar $\lambda$ let the sequence
$\lambda V:=[\lambda v_1,\lambda v_2,\ldots,\lambda v_{2t}]$.

\begin{figure}[h]
	\centering
	\includegraphics[scale=0.45]{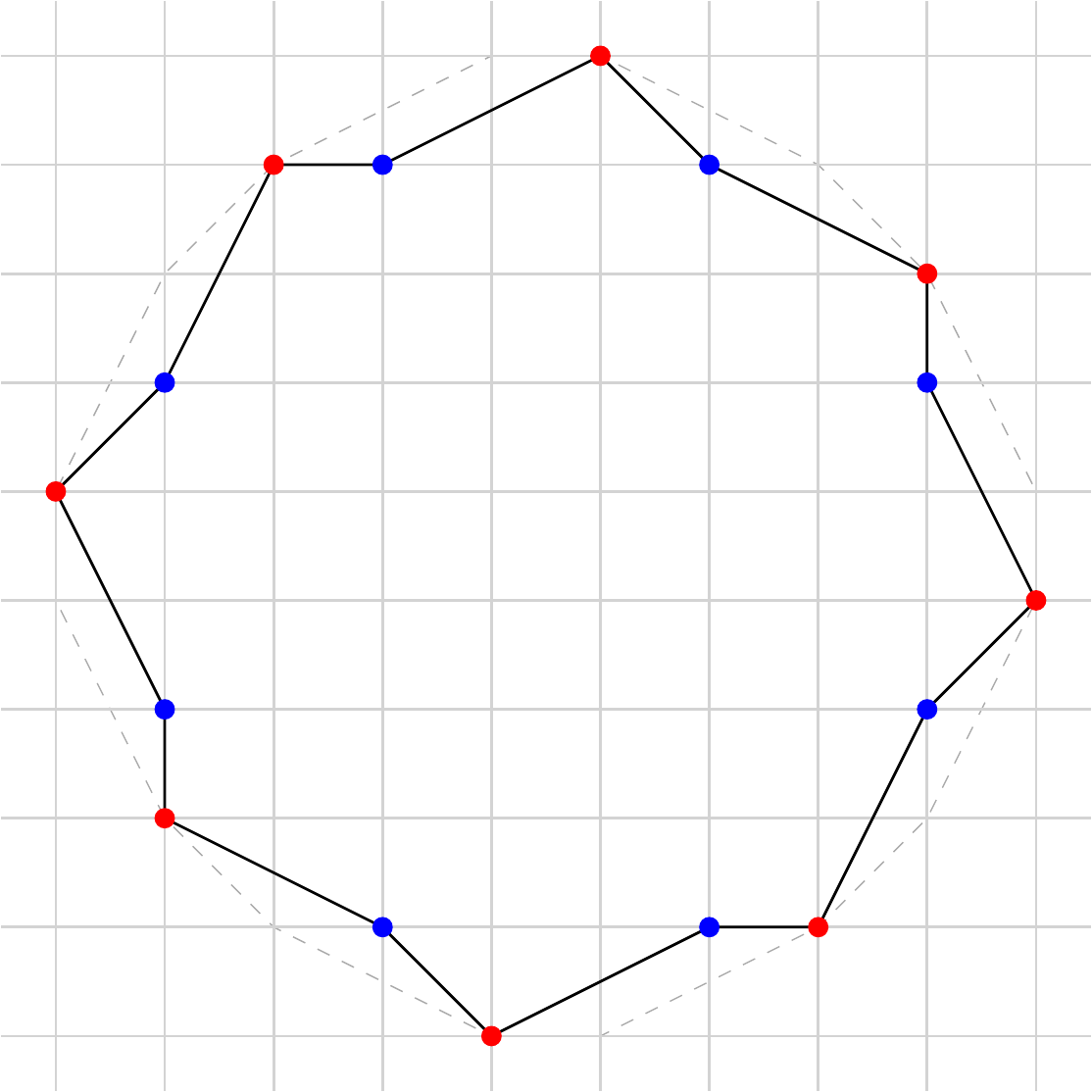}
	\caption{\small{$\mathcal{P}(\alt(V_4))$.}}
	\label{fig:alternate-example2}
\end{figure}

The idea is to generate a suitable symmetric sequence $V$ of $2n$ vectors and
then build the point set $\mathcal{P}(\alt(V))$ as the double circle point set,
up to some transformation of the elements of $\alt(V)$.
A (not optimal) example is $V=[(1,1),(1,2),\ldots,(1,n),(-1,-1),$ $(-2,-2),\ldots,(-1,n)]$
for even $n\ge 4$. The point set $\mathcal{P}(\alt(V))$ is in fact a double circle but its size is
equal to $1+2+\ldots+n=\Theta(n^2)$ (see Figure~\ref{fig:naive}).

\begin{figure}[h]
	\centering
	\begin{tabular}{cc}
		\includegraphics[scale=0.45]{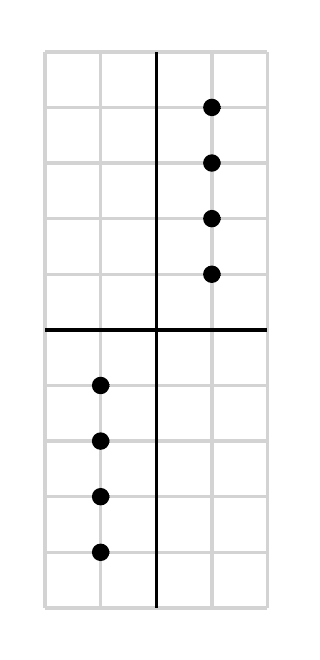} &
		\includegraphics[scale=0.45]{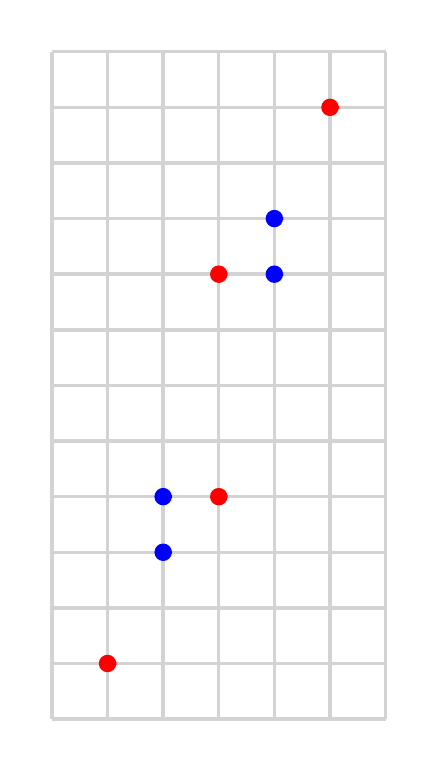}
	\end{tabular}
	\caption{\small{A naive construction for $n=4$ showing both vectors
		(left) and the resulting point set (right).}}	
	\label{fig:naive}
\end{figure}

The construction in which the resulting point set is a double circle of size $O(n^{3/2})$
is based on the next two algorithms:

{\sc VisibleVectors}$(n)$: With input $n\geq 3$, the symmetric sequence $V$ of
$2n$ visible vectors, sorted
counterclockwise around the origin, is generated so as to satisfy the next two invariants. Let
$\ball_t:=\{p\in\mathbb{Z}^2:\|p\|_1\leq t\}$, $k:=\max_{v\in V}\|v\|_1$, and
(even) $s$ be the number of visible vectors of $\ball_{k-1}$:
(i) all visible vectors of $\ball_{k-1}$ are in $V$, and
(ii) the other elements of $V$ are
generated as follows, until $2n-s$ elements are obtained: for $i=1,\dots,k-1$ generate
vectors $(i,k-i)$, $(-i,-(k-i))$, $(-i,k-i)$, $(i,-(k-i))$
in this order, if and only if $\gcd(i,k-i)=1$.
Refer to Algorithm~\ref{alg:visible-vectors} for
a pseudo-code.

{\sc BuildDoubleCircle}$(n)$: With input $n\geq 3$, build a
$2n$-point double circle. First, set sequences
$V:=${\sc VisibleVectors}$(n)$ and $[v'_1,v'_2,\ldots,v'_{2n}]:=\alt(V)$.
Then, the sequence $W=[w_1,w_2,\ldots,w_{2n}]$ of $2n$ vectors is created as follows:
for $i=1,3,\ldots,2n-1$ set
$w_i=(1-\lambda)v'_i+\lambda v'_{i+1}$ and
$w_{i+1}=\lambda v'_i+(1-\lambda) v'_{i+1}$, where $\lambda=1/3$. Finally, build the
$2n$-point set $\mathcal{P}((1/\lambda)W)$ as the double circle.

\begin{center}
\small
\begin{pseudocode}{VisibleVectors}{n}
\label{alg:visible-vectors}
k \GETS 1,
V \GETS [(1,0),(-1,0),(0,1),(0,-1)]\\
\REPEAT
	k\GETS k+1\\
	\FOR i \GETS 1 \TO k-1 \DO
	\BEGIN
		j \GETS k-i\\
		\IF \CALL{gcd}{i,j}=1 \THEN
		\BEGIN
			\IF \CALL{length}{V}<2n \THEN V\GETS V+[(i,j),(-i,-j)]\\
			\IF \CALL{length}{V}<2n \THEN V\GETS V+[(-i,j),(i,-j)]\\
		\END\\
	\END\\
\UNTIL \CALL{length}{V}=2n\\
\mbox{Sort } V \mbox{ counterclockwise around origin}\\
\RETURN V
\end{pseudocode}
\end{center}


\section{Construction correctness}\label{sec:proof}

Let $V=[v_1,v_2,\ldots,v_{2n}]$ be the (circular)
sequence of vectors obtained by executing
{\sc VisibleVectors}$(n)$, for $n \geq 3$.
For every $i=1,3,5,\ldots,2n-1$ we say that the pair
of vectors $v_{i},v_{i+1}$ is a pair of $\alt(V)$.

\begin{lemma}[Chapter 2 of~\cite{huxley1996}]\label{lem:jarnik-properties}
Given a natural number $Q$, the number $|V_Q|$ of vertices of the Jarn\'ik polygon is equal to
$$4+4\mbox{ }\underset{gcd(i,j)=1}{\sum\limits_{i=1}^Q\mbox{ }
\sum\limits_{j=1}^Q\mbox{ }1}=\frac{24Q^2}{\pi^2}+O(Q\log Q).$$
The size $S(Q)$ of the Jarn\'ik polygon is equal to
\begin{eqnarray*}
1+2\mbox{ }\underset{gcd(i,j)=1}{\sum\limits_{i=1}^Q\mbox{ }
\sum\limits_{j=1}^Q\mbox{ }i}&=&1+2\mbox{ }\underset{gcd(i,j)=1}{\sum\limits_{i=1}^Q\mbox{ }
\sum\limits_{j=1}^Q\mbox{ }j}\\
&=&\frac{6Q^3}{\pi^2}+O(Q^2\log Q).
\end{eqnarray*}
\end{lemma}

\begin{lemma}\label{lem:correct-visible-points}
$V$ is symmetric and point set $\mathcal{P}(V)$ has size $O(n^{3/2})$.
\end{lemma}

\begin{proof}
Observe that for every vector $a$ in $V$, $-a$ is also in $V$ since
in algorithm {\sc VisibleVectors} the vectors are added to sequence $V$ in pairs, and
each pair consists of two opposite vectors. Then $V$ becomes symmetric once
the elements of $V$ are sorted counterclockwise around the origin. On the other
hand $V_{\lfloor\frac{k-1}{2}\rfloor}\subset V\subset V_k$,
where $k=\max_{v\in V}\|v\|_1$.
Then we have
$|V_{\lfloor\frac{k-1}{2}\rfloor}|\leq 2n\le |V_k|$,
which implies $k=\Theta(\sqrt{n})$ by Lemma~\ref{lem:jarnik-properties}.
By the same lemma we obtain:
\begin{eqnarray*}
\sum_{i=1}^nx(v_i),\sum_{i=1}^ny(v_i)&<&1+2\mbox{ }\underset{gcd(i,j)=1}{\sum\limits_{i=1}^k\mbox{ }\sum\limits_{j=1}^k\mbox{}i}\\
&=&S(k)\\
&=&\Theta(k^3)=\Theta(n^{3/2})
\end{eqnarray*}
Hence, the size of $\mathcal{P}(V)$ is $O(n^{3/2})$.
\end{proof}

Let $o$ denote the origin of coordinates. Given two points $p,q$ let $\ell(p,q)$
denote the line passing through $p$ and $q$ and {\em directed} from $p$ to $q$, and
$pq$ denote the segment joining $p$ and $q$.
Given three points $p=(x_p,y_p)$, $q=(x_q,y_q)$, and $r=(x_r,y_r)$, let
$\Delta(p,q,r)$ denote the triangle with vertices at $p$, $q$, and $r$;
$A(p,q,r)$ denote de area of $\Delta(p,q,r)$; and
$\turn(p,q,r)$ denote the so-called {\em geometric turn} (going from $p$ to $r$ passing through $q$)
where
$$\turn(p,q,r)=\left|
\begin{array}{ccc}
x_p & y_p & 1\\
x_q & y_q & 1\\
x_r & y_r & 1
\end{array}
\right|$$
and
$A(p,q,r)=\frac{1}{2}\left|\turn(p,q,r)\right|$. Extending this notation,
let $\Delta(p,q):=\Delta(o,p,q)$, $A(p,q):=A(o,p,q)$, and $\turn(p,q):=\turn(o,p,q)$.
We use the so-called Pick's theorem:

\begin{theorem}[Pick's theorem~\cite{pick1899}]
The area of any simple
polygon $H$ with lattice vertices is equal to $i+b/2-1$, where $i$ and $b$ are the numbers of lattice
points in the interior and the boundary of $H$, 
respectively.
\end{theorem}

\begin{lemma}\label{lem:consecutive-vectors-1}
For every two consecutive vectors $a_1,a_2$ of $V$ we have $A(a_1,a_2)=1/2$.
\end{lemma}

\begin{proof}
Suppose $\Delta(a_1,a_2)$ contains a lattice point $p$ different from $o$, $a_1$, and $a_2$.
Then $p$ cannot belong to segments $oa_1$ and $oa_2$, and
segment $op$ contains a visible point $q$ (possibly equal to $p$).
If $\|q\|_1<\max\{\|a_1\|_1,\|a_2\|_1\}$ then $q$
must belong to $V$ by invariant (i)
of algorithm {\sc VisibleVectors}.
Otherwise, we have $\|q\|_1=\max\{\|a_1\|_1,\|a_2\|_1\}$.
Suppose w.l.o.g.\ that $\|a_1\|_1<\|a_2\|_1$, and 
let the point $q'$ denote the intersection of $\ell(o,q)$ with
the segment $s$ connecting $a_1$ to $a_2$. 
Observe that $q'=\delta a_1+(1-\delta)a_2$
for some $\delta\in(0,1)$, and further that 
$\|q\|_1\leq\|q'\|_1=\|\delta a_1+(1-\delta)a_2\|_1\leq \delta\|a_1\|_1+(1-\delta)\|a_2\|_2<\|a_2\|_1$,
which is a contradiction.
Then we must have that $\|q\|_1=\|a_1\|_1=\|a_2\|_1$,
which implies that $q,a_1,a_2$ belong to a same quadrant
since in this case $q$ is at the interior of the segment $s$.
Therefore, $q$ must belong to $V$ by invariant (ii) of algorithm {\sc VisibleVectors}.
In both cases, the fact that $q$ belongs to $V$ contradicts the fact that
$a_1$ and $a_2$ are consecutive vectors of $V$.
Hence $A(a_1,a_2)=1/2$ by Pick's theorem.\end{proof}

Let $\lambda\in(0,1/2)$. Given a pair $a,b$ of vectors
let $h(\lambda,a,b):=(1-\lambda)a+\lambda b$
(see Figure~\ref{fig:vector-tranf}).

\begin{figure}[h]
	\centering
	\includegraphics[width=4.0cm]{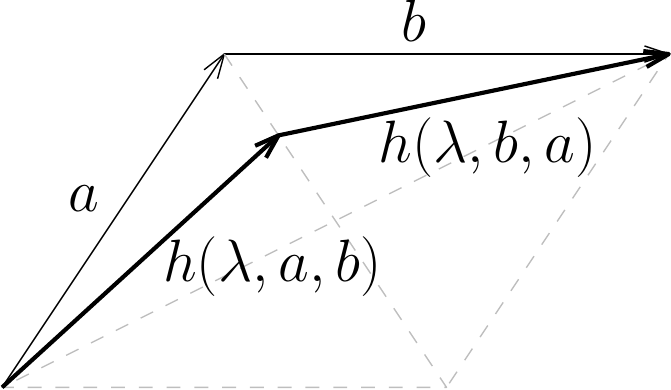}
	\caption{\small{
		Two vectors $a$ and $b$, and the
		vectors $h(\lambda,a,b)$ and $h(\lambda,b,a)$.}}	
	\label{fig:vector-tranf}
\end{figure}

\begin{lemma}\label{lem:lambda}
Let $a_1,a_2,a_3,a_4$ be four consecutive vectors of $V$ such that
$a_1,a_2$ and $a_3,a_4$ are pairs of $\alt(V)$.
Let $\lambda\in(0,1/2)$,
$q_1=h(\lambda,a_2,a_1)$, $q_2=q_1+h(\lambda,a_1,a_2)$, $q_3=q_2+h(\lambda,a_4,a_3)$,
and $q_4=q_3+h(\lambda,a_3,a_4)$. Then 
$q_2$ is to the right of $\ell(o,q_1)$ and both $q_3$ and $q_4$ are to the left of $\ell(o,q_1)$.
\end{lemma}

\begin{proof}
(Refer to Figure~\ref{fig:proof-lem-lambda}.)
We have that:
\begin{eqnarray*}
\turn(q_1,q_2) & = & \turn(q_1,q_1+h(\lambda,a_1,a_2))\\
 & = & \turn((1-\lambda)a_2+\lambda a_1,a_1+a_2)\\
 & = & (1-\lambda)\turn(a_2,a_1)+\lambda\turn(a_1,a_2)\\
 & = & 2(2\lambda-1)A(a_1,a_2)\\
 & < & 0,
\end{eqnarray*}
which implies that $q_2$ is to the right of the line $\ell(o,q_1)$. 
\begin{figure}[h]
	\centering
	\includegraphics[scale=0.8]{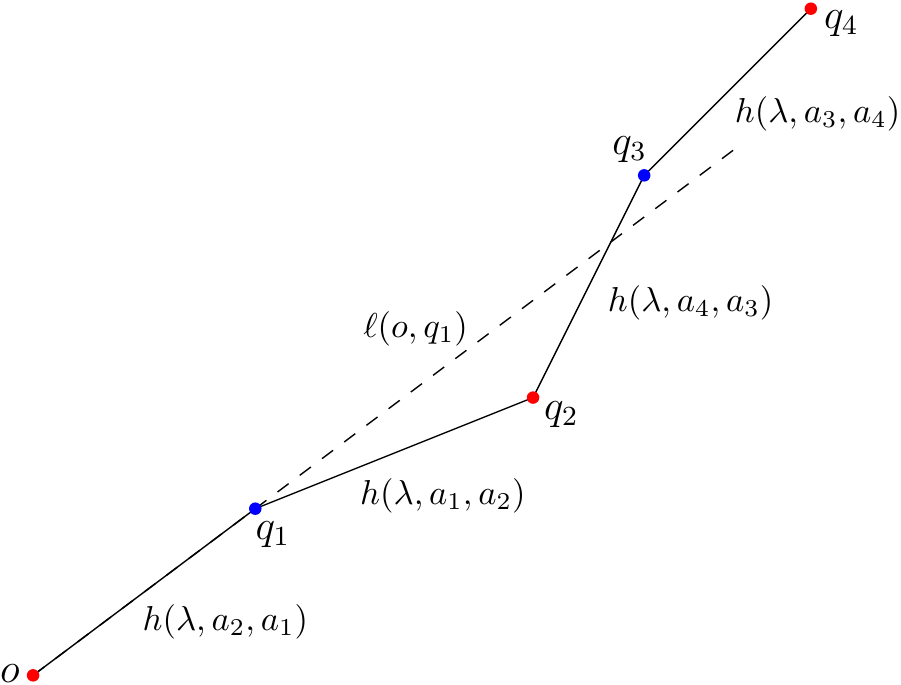}
	\caption{\small{Proof of Lemma~\ref{lem:lambda}.}}	
	\label{fig:proof-lem-lambda}
\end{figure}
On the other hand:
\begin{eqnarray}
\nonumber
\turn(q_1,q_3)&=&\turn(h(\lambda,a_2,a_1),h(\lambda,a_2,a_1)+h(\lambda,a_1,a_2)+\\
\nonumber
& &h(\lambda,a_4,a_3))\\
\nonumber
&=&\turn(h(\lambda,a_2,a_1),h(\lambda,a_1,a_2))+\\
\nonumber
& &\turn(h(\lambda,a_2,a_1),h(\lambda,a_4,a_3))\\
\nonumber
&=&\turn((1-\lambda)a_2+\lambda a_1,(1-\lambda)a_1+\lambda a_2))+\\
\nonumber
& &\turn((1-\lambda)a_2+\lambda a_1,(1-\lambda)a_4+\lambda a_3))\\
\nonumber
&=&(1-\lambda)^2\turn(a_2,a_1)+\lambda^2\turn(a_1,a_2)+\\
\nonumber
& &(1-\lambda)^2\turn(a_2,a_4)+\lambda(1-\lambda)\turn(a_2,a_3)+\\
\nonumber
& &\lambda(1-\lambda)\turn(a_1,a_4)+\lambda^2\turn(a_1,a_3)\\
\nonumber
&=& 2\Bigl((2\lambda-1)A(a_1,a_2)+(1-\lambda)^2A(a_2,a_4)+\\
\nonumber
& &\lambda(1-\lambda)A(a_2,a_3)+\lambda(1-\lambda)A(a_1,a_4)+\\
\nonumber
& &\lambda^2A(a_1,a_3)\Bigl)\\
\label{eq4}
&=& 2\biggl(\frac{1}{2}\left(2\lambda-1\right)+(1-\lambda)^2A(a_2,a_4)+\\
\nonumber
& &\lambda(1-\lambda)A(a_2,a_3)+\lambda(1-\lambda)A(a_1,a_4)+\\
\nonumber
& &\lambda^2A(a_1,a_3)\biggl)\\
\label{eq5}
&\ge& (2\lambda-1)+(1-\lambda)^2+\lambda(1-\lambda)+\\
\nonumber
& &\lambda(1-\lambda)+\lambda^2\\
\nonumber
&=&2\lambda>0
\end{eqnarray}
where equation~(\ref{eq4}) follows from Lemma~\ref{lem:consecutive-vectors-1} and
equation~(\ref{eq5}) follows from the fact that by Pick's theorem the area of any non-empty triangle with lattice
vertices is at least $1/2$. Therefore $q_3$ is to the left of $\ell(o,q_1)$.
Similarly, since we have that $\turn(a_i,a_j)>0$ ($i=1,2$; $j=3,4$) then
$$\turn(h(\lambda,a_2,a_1),h(\lambda,a_3,a_4))>0,$$ which implies that $q_4$ is to the left of $\ell(o,q_1)$
given that $q_3$ is to the left of $\ell(o,q_1)$.
By symmetry, it can be proved that $\turn(q_4,q_3,q_1)<0$ and $\turn(q_4,q_3,o)<0$,
implying that both $q_1$ and $o$ are to the right of $\ell(q_4,q_3)$.
\end{proof}

\begin{lemma}\label{lem:linear-time}
Algorithm {\sc VisiblePoints} can be implemented to run in $O(n)$ time
in the unit cost RAM model, $n\geq 3$.
\end{lemma}

\begin{proof}
Let $n\geq 3$ and $V$ be the answer of calling {\sc VisiblePoints}$(n)$.
Let $m:=1+\max_{v\in V}\|v\|_1$ which satisfies $m=\Theta(\sqrt{n})$ 
(see the proof of Lemma~\ref{lem:correct-visible-points}). 
Using $O(m^2)$ space, and the facts $\gcd(i,i)=i$ and
$\gcd(i,j)=\gcd(i-j,j)$ for $i>j$, one can compute
$\gcd(i,j)$ in constant time for any $i,j$. Then, computing
$V$ without the radial sorting around the origin requires $O(m^2)=O(n)$ time.

We show now that $V$ can be radially sorted around the origin in $O(n)$ 
time via Bucket Sort,
and it suffices to show how to sort the vectors $V'\subset V$ that
belong to the interior of the first quadrant. Let $b_t$ denote the 
point $(t/m,m-t/m)$ for $t=0,1,\ldots,m^2$, and consider the 
$m^2$ triangles $\Delta_1,\Delta_2,\ldots,\Delta_{m^2}$ 
as buckets, where triangle $\Delta_t$ $(t=1,\ldots,m^2)$ has vertices $o$, $b_{t-1}$, and
$b_t$. Observe that triangles $\Delta_1,\Delta_2,\ldots,\Delta_{m^2}$ have
pairwise disjoint interiors, and all have area equal 
to $1/2$. If $a_1$ and $a_2$ are
two different vectors of $V'$ then $A(o,a_1,a_2)\ge 1/2$, which implies that
$a_1$ and $a_2$ cannot belong to a same triangle (bucket) $\Delta_t$
given that both are not contained in the segment $b_{t-1}b_t$ for $t=1,\ldots,m^2$. Therefore,
every triangle among $\Delta_1,\Delta_2,\ldots,\Delta_{m^2}$ contains
at most one point of $V'$. Given any vector $a:=(i,j)\in V'$, the triangle
$\Delta_t$ that contains $(i,j)$ can be found in constant time. Namely,
$t$ is the smallest value in the range $[1\ldots m^2]$ such that
$\turn(o,a,b_t)\le 0$, that is,
\begin{eqnarray*}
\left|
\begin{array}{ccc}
0 & 0 & 1\\
i & j & 1\\
t/m & m-t/m & 1
\end{array}
\right| =  i(m-t/m)-j(t/m) &\leq & 0\\
\frac{i\cdot m^2}{i+j} & \leq & t
\end{eqnarray*}
where $t$ satisfies $t=\lceil\frac{i m^2}{i+j}\rceil$. Since $t$, and then $\Delta_t$, can be
found in constant time, the vectors of $V'$ can be sorted in
$O(m^2)=O(n)$ time.
\end{proof}

\begin{theorem}\label{teo:double-circle}
There is an $O(n)$-time algorithm that for all $n\ge 3$ builds a double circle of $2n$ points
in the grid $[0,N]^2$ where $N=O(n^{3/2})$.
\end{theorem}

\begin{proof}
Execute the algorithm {\sc BuildDoubleCircle} with input $n$,
being $V$ the result of calling {\sc VisiblePoints}$(n)$,
building the point set $P$ of $2n$ points. Observe that $\lambda=1/3$
implies that point $w_i/\lambda=3 w_i$ is integer for $i=1\ldots 2n$, and then all elements of $P$
are integer points. By Lemma~\ref{lem:lambda}, the point set $P$
is a double circle.
The size of $\mathcal{P}(V)$ is $O(n^{3/2})$ by Lemma~\ref{lem:correct-visible-points},
and since all elements of $P$ belong to the polygon with vertices $\mathcal{P}(3V)$,
the size $N$ of $P$ is also $O(n^{3/2})$.
Finally, translate $P$ to lie in the grid $[0,N]^2$. 
Since algorithm {\sc VisiblePoints} can run in $O(n)$ time (Lemma~\ref{lem:linear-time}),
algorithm {\sc BuildDoubleCircle} can be done in $O(n)$ time. 
The result follows.
\end{proof}

\section{Future work}\label{sec:conclusions}

We are working on extending the results of this paper to build other known point
sets in integer points of small size, such as the double convex chain, the Horton set, and others.
We plan to eventually release a software library supporting many of these constructions.

\small

%

\small

\bibliographystyle{plain}
\bibliography{double_circle}


\newpage
\appendix

\section{Examples}\label{sec:examples}

\begin{figure}[h]
	\centering
	\subfloat[]{
		\begin{tabular}{cc}
			\includegraphics[scale=0.4]{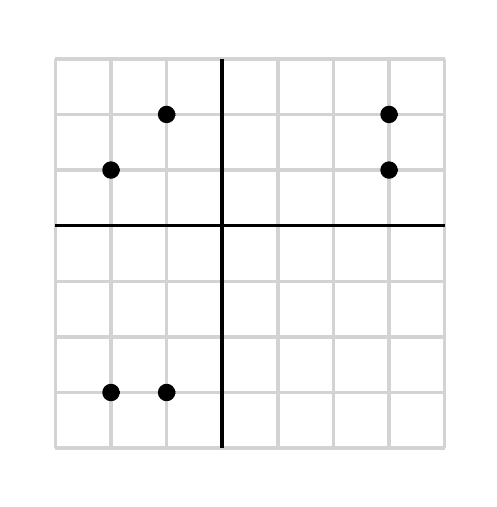} &
			\includegraphics[scale=0.4]{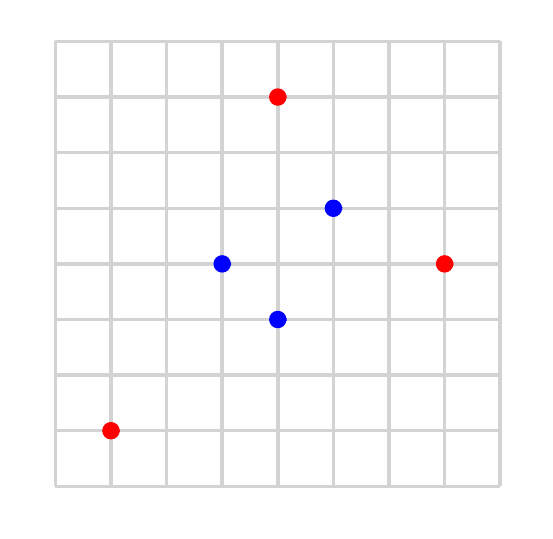}
		\end{tabular}
		\label{fig:n3}
	}
	\subfloat[]{
		\begin{tabular}{cc}
			\includegraphics[scale=0.4]{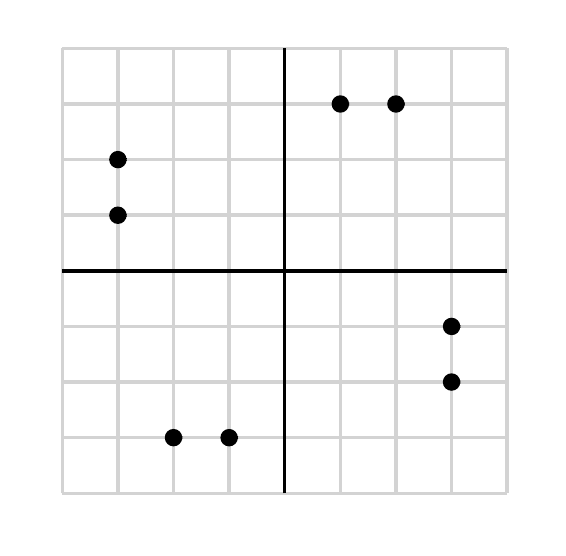} &
			\includegraphics[scale=0.4]{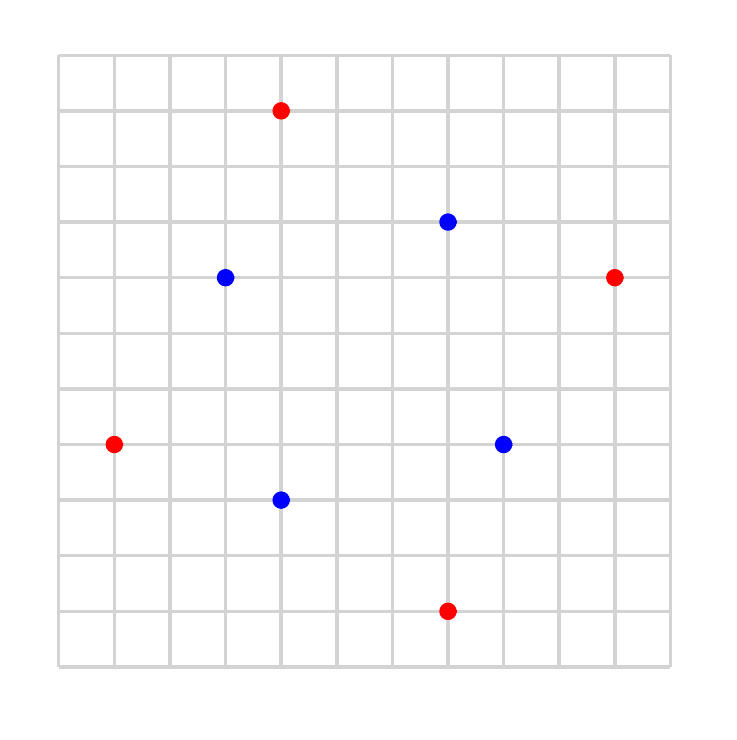}
		\end{tabular}
		\label{fig:n4}
	}\\
	\subfloat[]{
		\begin{tabular}{cc}
			\includegraphics[scale=0.35]{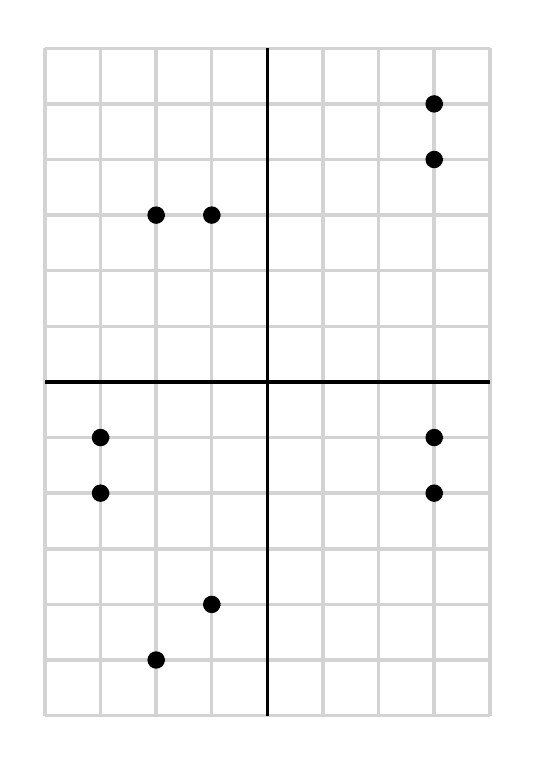} &
			\includegraphics[scale=0.35]{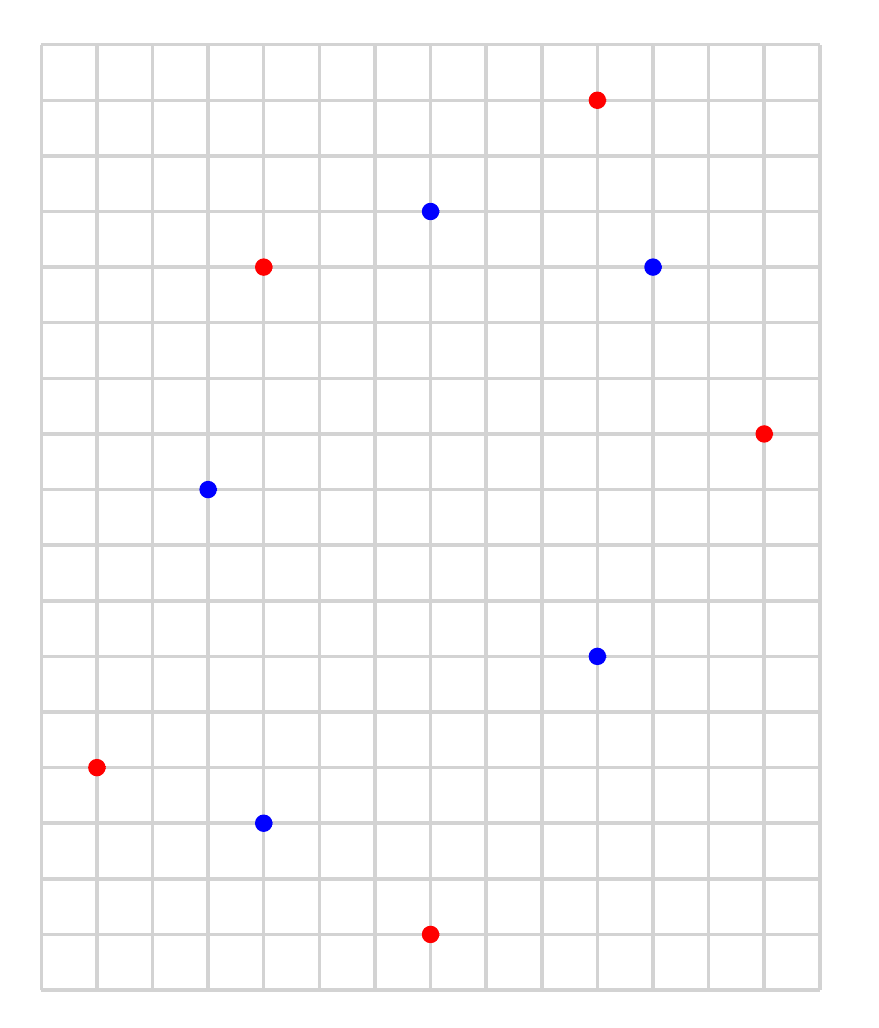}
		\end{tabular}
		\label{fig:n5}
	}\\
	\subfloat[]{
		\begin{tabular}{cc}
			\includegraphics[scale=0.35]{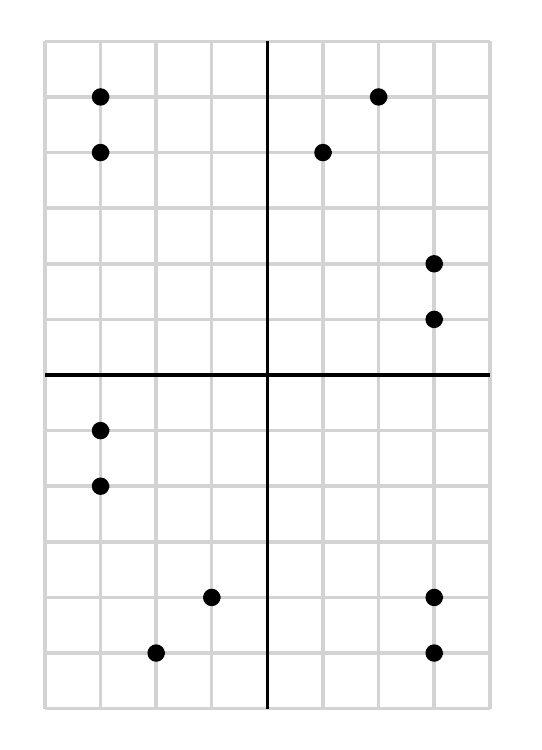} &
			\includegraphics[scale=0.35]{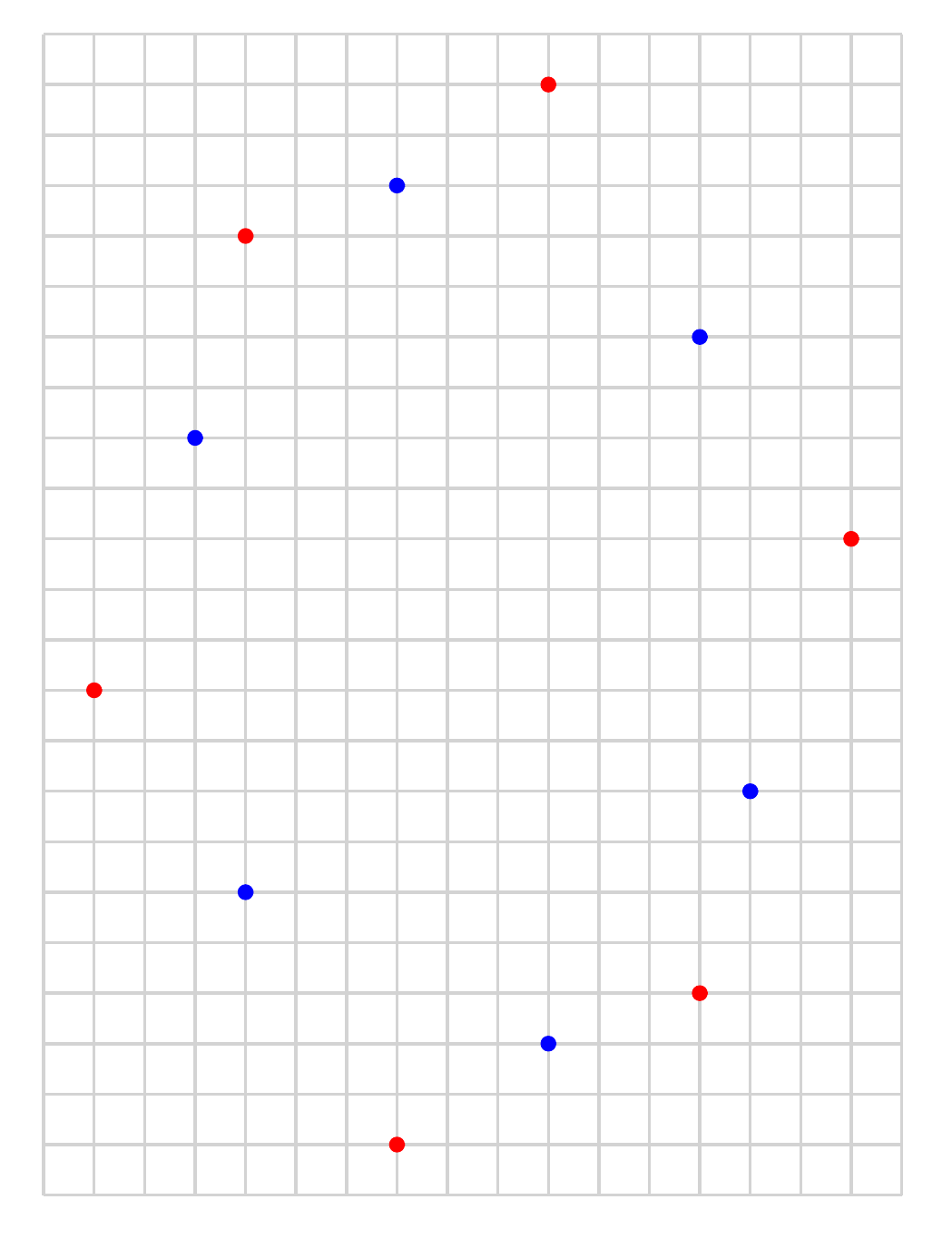}
		\end{tabular}
		\label{fig:n6}
	}
	\caption{\small{Output of the algorithm for $n=3,4,5,6$ (right) and the corresponding vectors $W$ (left).
	}}	
\end{figure}


\begin{figure}[h]
	\centering
	\subfloat[]{
		\includegraphics[scale=0.15]{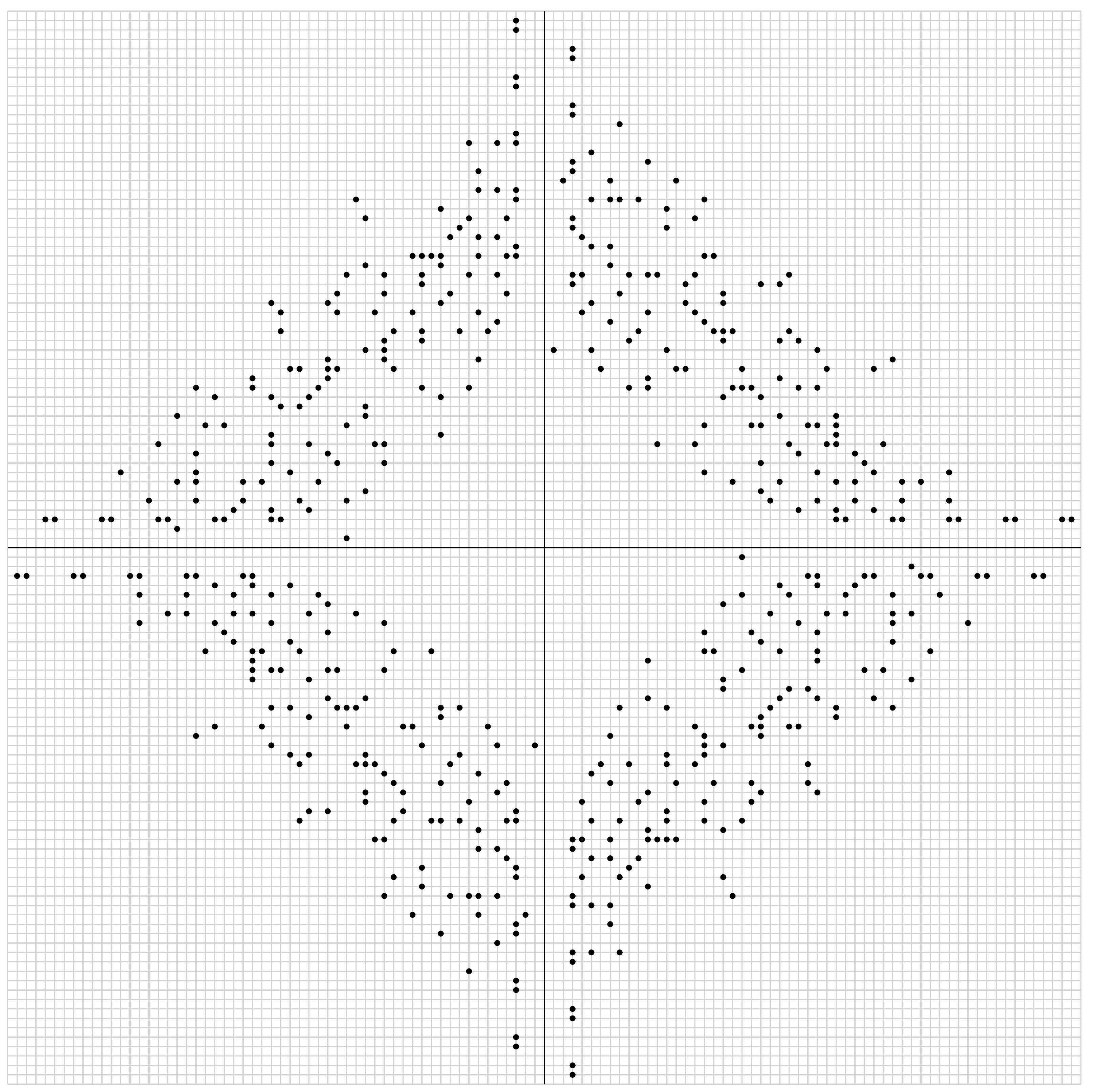}
		\label{fig:n256}
	}
	\subfloat[]{
		\includegraphics[scale=0.25]{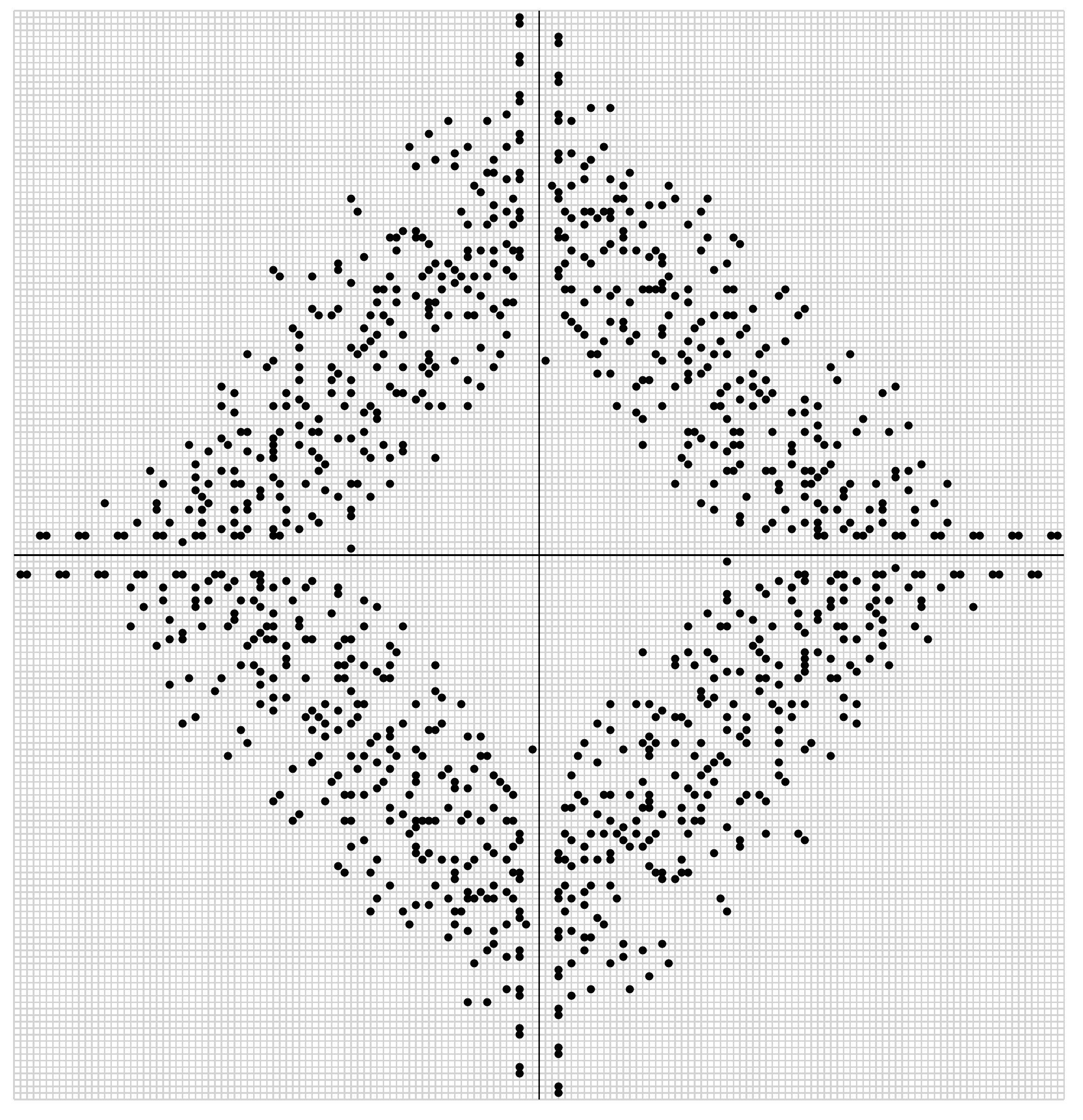}
		\label{fig:n512}
	}\\
	\subfloat[]{
		\includegraphics[scale=0.25]{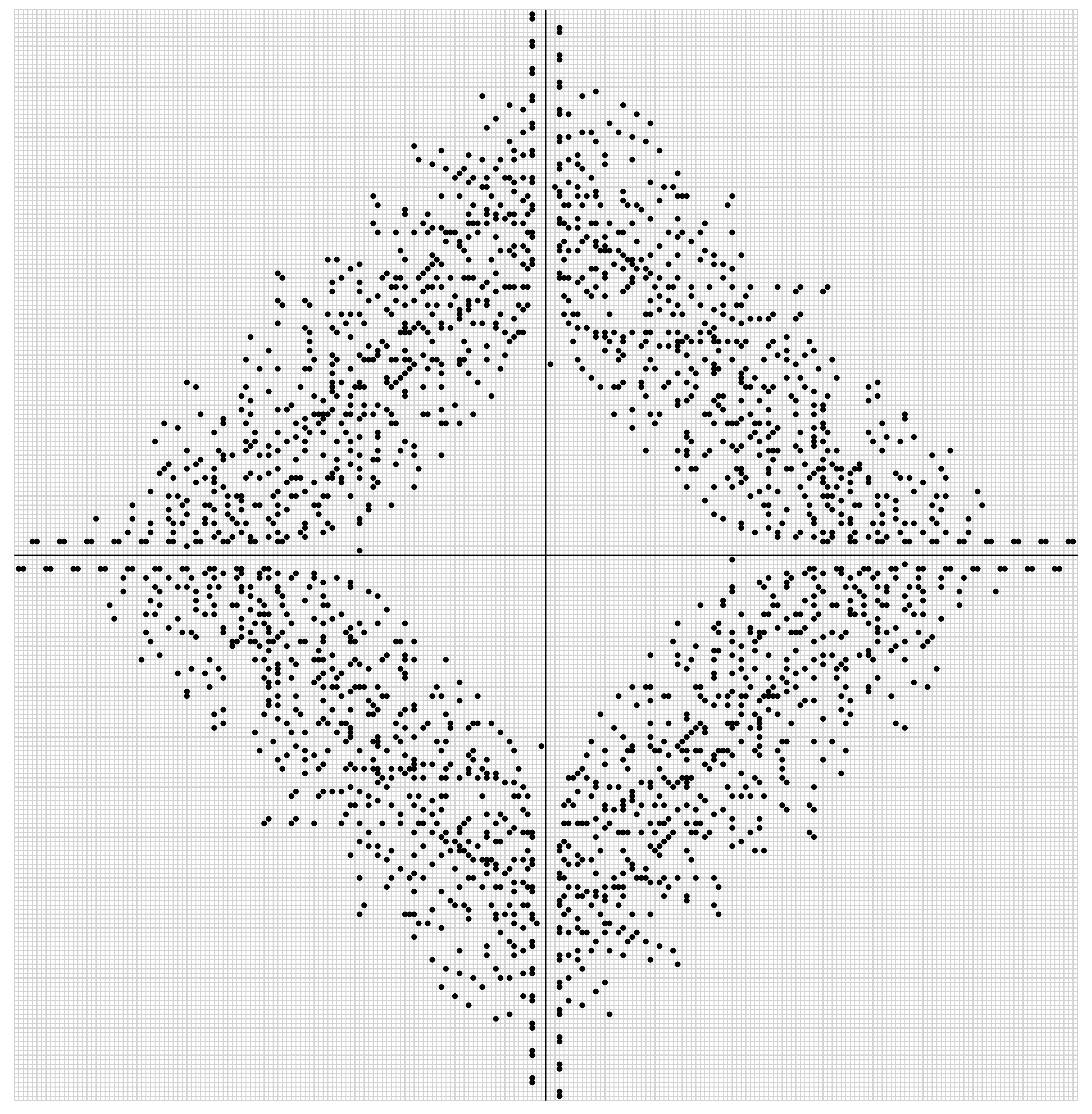}
		\label{fig:n1024}
	}
	\caption{\small{Vectors of sequence $W$ for $n=256,512,1024$.}}	
\end{figure}

\end{document}